\documentclass[12pt]{article}

\usepackage[pdftex]{graphicx}
\usepackage{sidecap}
\usepackage{amsthm}
\usepackage{epsfig,color}

\usepackage{framed}
\newcommand{\comment}[1]{}
\newcommand{\QED}{\mbox{}\hfill \rule{3pt}{8pt}\vspace{10pt}\par}

\newcommand{\etal}{\textit{et al.}}

\usepackage{booktabs} 
\usepackage{algorithm,algpseudocode}
\usepackage[cmex10]{amsmath}
\usepackage{amssymb,mathtools}
\usepackage{color}
\usepackage[table, dvipsnames]{xcolor}
\usepackage{graphicx}
\usepackage{subfig}
\usepackage{wrapfig}
\usepackage{verbatim} 
\usepackage{enumerate}
\usepackage{pifont}
\usepackage{hyperref}

\newtheorem{theorem}{Theorem}[section]
\newtheorem{lemma}[theorem]{Lemma}

\newtheorem{definition}{Definition}

\newtheorem{remark}[theorem]{Remark}

\newcommand{{\rh}}{{\widehat r}}
\newcommand{{\Rh}}{{\widehat R}}

\usepackage{mathtools}

\newcommand{\anis}[1]{{\color{green}\underline{\textsf{A:}}} {\color{blue} \emph{#1}}}
\newcommand{\manish}[1]{{\color{blue}\underline{\textsf{M:}}} {\color{red} \emph{#1}}}

\newboolean{short}
\setboolean{short}{false}
\newcommand{\shortOnly}[1]{\ifthenelse{\boolean{short}}{#1}{}}
\newcommand{\onlyShort}[1]{\ifthenelse{\boolean{short}}{#1}{}}
\newcommand{\longOnly}[1]{\ifthenelse{\boolean{short}}{}{#1}}
\newcommand{\onlyLong}[1]{\ifthenelse{\boolean{short}}{}{#1}}
\newcommand{\shortLong}[2]{\ifthenelse{\boolean{short}}{#2}{#1}}
\newcommand{\longShort}[2]{\ifthenelse{\boolean{short}}{#2}{#1}} 
\onlyShort{

\usepackage{enumitem}
\setitemize{noitemsep,topsep=0pt,parsep=0pt,partopsep=0pt}
\setenumerate{noitemsep,topsep=0pt,parsep=0pt,partopsep=0pt}
}
    
\begin{document}

%

\title{Fault-Tolerant Dispersion of Mobile Robots}

\author{Prabhat Kumar Chand \thanks{Indian Statistical Institute, Kolkata 700108, India.\hbox{E-mail}:~{\tt pchand744@gmail.com}.}
\and Manish Kumar \thanks{Indian Statistical Institute, Kolkata 700108, India.\hbox{E-mail}:~{\tt manishsky27@gmail.com}.}
\and Anisur Rahaman Molla \thanks{Indian Statistical Institute, Kolkata 700108, India.  \hbox{E-mail}:~{\tt anisurpm@gmail.com}.} \and Sumathi Sivasubramaniam \thanks{Indian Statistical Institute, Kolkata 700108, India. \hbox{E-mail}:~{\tt sumathivel89@gmail.com}.}}

%


\date{}

\maketitle \thispagestyle{empty}

\maketitle

\begin{abstract}
We consider the mobile robot dispersion problem in the presence of faulty robots (crash-fault). Mobile robot dispersion consists of $k\leq n$ robots in an $n$-node anonymous graph. The goal is to ensure that regardless of the initial placement of the robots over the nodes, the final configuration consists of having at most one robot at each node. In a crash-fault setting, up to $f \leq k$ robots may fail by crashing arbitrarily and subsequently lose all the information stored at the robots, rendering them unable to communicate. In this paper, we solve the dispersion problem in a crash-fault setting by considering two different initial configurations: i) the rooted configuration, and ii) the arbitrary configuration. In the rooted case, all robots are placed together at a single node at the start. The arbitrary configuration is a general configuration (a.k.a. arbitrary configuration in the literature) where the robots are placed in some $l<k$ clusters arbitrarily across the graph. For the first case, we develop an algorithm solving dispersion in the presence of faulty robots in $O(k^2)$ rounds, which improves over the previous $O(f\cdot\text{min}(m,k\Delta))$-round result by \cite{PS021}. For the arbitrary configuration, we present an algorithm solving dispersion in $O((f+l)\cdot\text{min}(m, k \Delta, k^2))$ rounds, when the number of edges $m$ and the maximum degree $\Delta$ of the graph is known to the robots.
\end{abstract}

%
%
%


\section{Introduction}\label{sec:introduction}
The dispersion of autonomous mobile robots to spread them out evenly in a region is a problem of significant interest in distributed robotics, e.g. \cite{HABFM03, HABFM02}. Initially, this problem was formulated by Augustine and Moses Jr. \cite{AM18} in the context of graphs. They defined the problem as follows: Given any arbitrary initial configuration of $k \leq n$ robots positioned on the nodes of an $n$-node anonymous graph, the robots reposition autonomously to reach a configuration where each robot is positioned on a distinct node of the graph. Mobile robot dispersion has various real-world and practical applications, such as the relocation of self-driving electric cars (robots) to recharge stations (nodes). Assuming that the cars have smart devices to communicate with each other, the process to find a free or empty charging station, coordination including exploration (to visit each node of the graph in minimum possible time), scattering (spread out in an equidistant manner in symmetric graphs like rings), load balancing (nodes send or receives loads, and distributes them evenly among the nodes), covering, and self-deployment can all be explored as mobile robot dispersion problems. \cite{KA19,KMS19, KMS020, KMS22}. 

The problem has been extensively studied in different graphs with varying assumptions since its conceptualisation \cite{MollaM19,KA19,KMS020, KMS19, KMS20, KMS22, KS21, MM22, MMM020, MMM21}. In this paper, we continue the study about the trade-off of memory requirement and time to solve the dispersion problem.  Recently, Pattanayak et al. \cite{PS021} explored the problem of dispersion in a set-up where some of these mobile robots are prone to crash faults. Whenever a robot crashes, it loses all its information immediately, as if the robot has vanished from the network. This makes the problem more challenging and also makes the problem more realistic in terms of real world scenarios, where faulty robots can crash at any moment. In this paper, we have continued to study the efficacy of the problem in the same faulty environment. We have studied the dispersion problems with the rooted and arbitrary configuration of the robots with faulty setup. Both the algorithms maintain optimal level of memory requirement for each robot.


In our work, we mainly discuss two problems of mobile robot dispersion in the presence of crash-fault: i) rooted configuration and ii) arbitrary configuration. In the rooted case, all robots are placed together at a single node (called the $root$) at the start. On the other hand, the arbitrary configuration is a general configuration where the robots are initially placed in some $l<k$ clusters arbitrarily across the graph. Our first algorithm works for the rooted configuration by using depth first search (DFS) traversal and improves the round complexity from $O(f\cdot\text{min}(m,k\Delta))$ \cite{PS021} rounds to $O(k^2)$. The second algorithm for the arbitrary graph is an entirely new result whose complexity depends upon the factors like the number of faulty robots $(f)$, number of robot clusters $(l)$, total number of edges in the graph $(m)$, number of robots $(k)$ and the highest degree of the graph $(\Delta)$. In this case, we have a round complexity of  $O((f+l)\cdot \text{min}(m, k \Delta, k^2))$.

\subsection{Notations at a Glance}

\begin{center}
\begin{tabular}{ p{2cm}|p{10cm} }
    \hline
    \multicolumn{2}{c}{
    \textbf{Notations used throughout the paper}} \\
    \hline
    \centering \textbf{Symbols} & \textbf{Meaning} \\
    \hline
    \centering $G$ & The arbitrary graph acting the underlying network for the robots\\
    \hline
    \centering $n$ & The number of  nodes(vertices) of $G$\\
    \hline
    \centering $m$ & The number of edges of $G$\\
    \hline
    \centering $\Delta$ & The highest degree among the nodes of $G$\\
    \hline
    \centering $k$ & Number of robots\\
    \hline
    \centering $f$ & Number of faulty robots among $k$\\
    \hline
    \centering $l$ & Number of initial clusters of robots in the $clustered$ configuration\\
    \hline
    \centering $r_i$ & A robot with ID $i$\\
    \hline
    \centering $R_c$ & $root$ vertex in the $rooted$ configuration\\
    \hline
\end{tabular}
\end{center}
\vspace{-0.5cm}

\subsection{Challenges and Techniques} 


Both of our Algorithms, \ref{alg: rooted_config} and \ref{alg: cluster_config} have Depth First Search as their foundation. Since the nodes themselves are indistinguishable, navigation is done via the nodes that settle on the nodes.  Furthermore, the robots cannot communicate between themselves unless they are in the same node. In a crash setting, this causes immediate problems in navigation. Faulty robots, when they crash at inappropriate times, can also create endless cycles or can increase the number of clusters on the arbitrary initial configuration. In this paper, we have tried to solve the problem by overcoming the said challenges, minimising the time complexity and keeping an optimal memory requirement for each robot.

For the rooted initial configuration, we perform a DFS search on $G$ from the $root$ vertex. Robots from the root are released one by one as they explore the graph sequentially. The first robot from the root settles down at the root and sets the minimum port number (that is yet unexplored) available at the root as its current direction pointer ($cdr$). The next robot follows the $cdr$ of the previous robot and reaches a new node, where it settles and sets its own $cdr$ pointer. Continuing in a similar way, the succeeding robots build upon the DFS and continue to explore the graph. While exploring the graph, the robots can complete the exploration of a certain part of the graph via a particular vertex, when, we set the backtrack value of the robot in the particular vertex to $1$. This ensures no further robots visit this part of the graph unnecessarily. In our algorithm, the $i^{th}$ robot is sent only after $3i$ rounds have elapsed. The $i$ robot uses the $3i$ rounds to explore and if needed, return to the root (the robot returns to the root if it does not find a new node to settle after $2i$ rounds). In such a case, the robot keeps exploring new edges before it breaks out into a new path and settles at a new node there. During the algorithm, if any robot crashes, the succeeding robots from the root correct any inconsistency in the pointers of the settled robots. We claim that each such crash could only extend the number of rounds by an $O(k)$, at-most (Lemma \ref{lem: extra_cost}) and eventually the algorithm (Algorithm \ref{alg: rooted_config}) completes in $O(k^2)$ rounds.

The arbitrary configuration required a different approach, however. At the start, we have $l< k$ clusters of robots at $l$ different nodes of the graph. Our algorithm (Algorithm \ref{alg: cluster_config}) runs in phases, where each phase consists of $\text{min}(k\Delta, k^2)$ rounds. At the start of each phase, each cluster begins a $counter$ that counts down from $\text{min}(k\Delta, k^2)$. At the start of the phase, each cluster $C_i$ begins exploring the network in parallel using a DFS approach.  Unlike the rooted configuration, individual robots do not explore and return, but the entire cluster moves together. Whenever a cluster encounters a new (empty) node in the network, the robot with the current highest ID in the cluster settles, and updates its flag variables accordingly. Each time the cluster moves through an edge to a different node, the counter is decreased by 1. When the counter becomes zero, all flags are reset. After that, each cluster starts exploring the network with its current node as a point of origin. This continues until all robots in the cluster settle or the algorithm ends. The algorithm completes in $O((l+f)min(k\Delta,k^2))$ rounds with each robot requiring an optimal $O(\log (k+\Delta))$ bits of memory.

\section{Related Work}\label{sec:relatedwork}
The problem of dispersion was first introduced in \cite{AM18} by Moses Jr. \etal, where they solved the problem for different types of graphs. They had given a lower bound of $\Omega(\log n)$ on the memory of each robot (later, made more specific with $\Omega(\log(\text{max}(k, \Delta)))$ in \cite{KMS19}) and of $\Omega(D)$ on the time complexity, for any deterministic algorithm on arbitrary graphs. They also proposed two algorithms on arbitrary graphs, one requiring $O(\log n)$ memory and running for $O(mn)$ time while the other needing a $O(n\log n)$ memory and having a time complexity of $O(m)$ .

Kshemkalyani and Ali~\cite{KA19} provided several algorithms for both synchronous and asynchronous models. In the synchronous model, they solved the dispersion problem in $O(\min(m,k\Delta))$ rounds with $O(k\log \Delta)$ memory. For the asynchronous cases, they proposed several algorithms, one particularly requiring $O(\Delta^D)$ rounds and $O(D\log \Delta)$ memory, while another requiring $O(\max(\log k, \log \Delta))$ memory and having a time complexity of $O((m-n)k)$. In a later work, Kshemkalyani \etal, in \cite{KMS19} improved the time complexity to $O(\min(m,k\Delta)\log k)$ keeping the memory requirement to $O(\log n)$, while requiring that the robots know the parameters $m,n,k,\Delta$ beforehand. In subsequent work, \cite{SSKM20} kept the time and memory complexity of \cite{KMS19} intact while dropping the requirement of the robots to have prior knowledge of $m,k,\Delta$. Recently, Kshemkalyani and Sharma \cite{KS21} improved the time complexity to $O(\min(m,k\Delta))$. Works of \cite{MollaM19} and \cite{DBS21} used randomisation, which helped to reduce the memory requirement for each robot.

In \cite{KMS20}, Kshemkalyani \etal, studied the problem in the \textit{Global Communication Model}, in which the robots can communicate with each other irrespective of their positions in the graph\footnote{In the \textit{Local Communication Model} robots can communicate with each other only when they are at the same node.}. The authors obtained a time complexity of $O(k\Delta)$ rounds when $O(\log (k+\Delta))$ bits of memory were allowed at each robot.  Whereas, when robots were allowed $O(\Delta+\log k))$ bits, the number of rounds reduced to $O(\min(m, k\Delta))$. Both were for arbitrary initial configuration of robots. They also used BFS traversal techniques for investigating the dispersion problem. The BFS traversal technique yielded a time of $O((D+k)\Delta(D+\Delta))$ rounds with $O(\log D + \Delta\log k)$ bits of memory at each robot, using global communication, for arbitrary starting configuration of robots. Here $D$ denotes the diameter of the graph. The problem was also studied on \textit{dynamic} graphs in \cite{KMS020},\cite{AAMKS18},\cite{LFPPSV20}. \textit{Graph Exploration}, which is a related problem, has also been intensively studied in literature \cite{BGHIKK09} \cite{CFIKP08} \cite{DDKPU13} \cite{FIPPP05}

The dispersion problem has also been recently studied for configurations with faulty robots. In\cite{MMM20}, Molla \etal, considered the problem for anonymous rings, tolerating weak Byzantine faults (robots that behave arbitrarily but cannot change their IDs). They gave three algorithms \textbf{(i)} the first one being memory optimized, requiring $O(\log n)$ bits of memory, $O(n^2)$ rounds and tolerating up-to $n-1$ faults.\textbf{(ii)} the second one is time optimized with $O(n)$ rounds, but require $O(n\log n)$ bits of memory, tolerating up-to $n-1$ faults. \textbf{(iii)} the third one runs in $O(n)$ time and $O(\log n)$ memory but cannot tolerate more than $[\frac{n-4}{17}]$ faulty robots. In \cite{MMM21}, the authors proposed several algorithms for dispersion with some of them tolerating strong Byzantine robots (robots that behave arbitrarily and can tweak their IDs as well). Their algorithms are mainly based on the idea of gathering the robots at a root vertex, using them to construct an isomorphic map of $G$ and finally dispersing them over $G$ according to a specific protocol. However, their algorithms take exponential rounds for strong Byzantine robots starting from an arbitrary configuration. For the rooted configuration, their algorithm takes $O(n^3)$ rounds, but tolerates no more than $[n/4-1]$ strong Byzantine robots. Dispersion under \textit{crash faults} has been dealt with in \cite{PS021}. In \cite{PS021}, Pattanayak \etal, have considered the problem for a team of robots starting at a rooted configuration, with some robots being crash prone. Their algorithm handles an arbitrary number of crashes, with each robot requiring $O(\log(k+\Delta))$ bits of memory. The algorithm completes in $O(f\cdot \min(m,k\Delta))$ rounds. In our paper, we improve this time complexity while keeping the memory requirement to optimal and also extend the problem for the robots starting in arbitrary configuration. A comparison between our results and the most aligned works is shown in Table~\ref{tab:results}.

\begin{center}
\begin{table}[H]
    \centering
    \begin{tabular}{ p{4.5cm} p{2cm} p{2cm} p{4.5cm} }
    \hline
    \textbf{Algorithm}  & \textbf{Initial Config.} & \textbf{Crash Handling} & \textbf{Time}\\
    \hline
    Kshemkalyani \etal \cite{KS21}* & Arbitrary & No & $O(\min(m,k\Delta))$ \\
    Pattanayak \etal \cite{PS021} & Rooted & Yes & $O(f \cdot \min(m, k\Delta) )$ \\
    \textbf{Algorithm in Sec.~\ref{sec: rooted_configuration} }& Rooted & Yes & $O(k^2)$\\
    \textbf{Algorithm in Sec.~\ref{sec: cluster} }& Arbitrary & Yes & $O((f+l) \cdot
    \min(m,k\Delta, k^2))$\\
    
\hline
\end{tabular}
    \caption{Results on Dispersion of $k \leq n$ robots with $f\leq k$ faulty robots on $n$-node arbitrary anonymous graphs having $m$ edges such that $\Delta$ is the highest degree of the graph in the local communication model. Each uses an optimal memory of $O(\log(k+\Delta))$ bits on each robot. \centering *\textit{The best known result as of now for fault-free dispersion.}}
    \label{tab:results}
\end{table}
\end{center}

\vspace{-1.2cm}
\section{Model and Problem Definition}\label{sec:model}
We now elaborate on our model and problem in detail.\\\\
\textbf{Graph: } The underline graph $G$ is connected, undirected, unweighted and anonymous with $|V| = n$ vertices and $|E| = m$ edges. The vertices of $G$ (also called nodes) do not have any distinguishing identifiers or labels. 
The nodes do not possess any memory and hence cannot store any information. The degree of a node $i\in V$ is denoted by $\delta_i$ and the maximum degree of $G$ is $\Delta$. Edges incident on $i$  are locally labelled using a port number in the range $[1,\delta_i]$. A single edge connecting two nodes receives two independent port numbers at either end. The edges of the graph serve as $routes$ through which the robots can commute. Any number of $robots$ can travel through an edge at any given time.\\\\
\textbf{Robots: }We have a collection of $k\leq n$ robots $\mathbb{R} = \{r_1,r_2,...,r_k\}$ residing on the nodes of the graph. Each robot has a unique ID and has some memory to store information. The robots cannot stay over an edge, but one or more robots can be present at a node at any point of time. A group of robots at a node is called $co-located$ robots. Each robot knows the port number through which it has entered and exited a node.\\\\
\textbf{Crash Faults: }The robots are not fault-proof and a faulty robot can $crash$ at any time during the execution of the algorithm. Such $crashes$ are not recoverable and once a robot $crashes$ it immediately loses all the information stored in itself, as if it was not present at all. Further, a crashed robot is not visible or sensible to other robots. We assume there are $f$ faulty robots such that $f\leq k$.\\\\
\textbf{Communication Model:} Our paper considers a local communication model where only the co-located robots can communicate among themselves.\\ 

\noindent\textbf{Time Cycle: } Each robot $r_i$, on activation, performs a $Communicate-Compute-Move$ $(CCM)$ cycle as follows.
\vspace{-0.1cm}
    \begin{itemize}
        \item Communicate: $r_i$ reads its own memory along with the memory of other robots co-located at a node $v_i$.
        \item Compute: Based on the gathered information and subsequent computations, $r_i$ decides on several parameters. This includes, deciding whether to settle at $v_i$ or otherwise determine an appropriate exit port, choosing the information to pass/store at the settled robot and the information to carry along-with, if, exiting $v_i$.
        \item Move: $r_i$ moves to the neighboring node using the computed exit port.
    \end{itemize}
\vspace{-0.3cm}
We consider a synchronous system, where every robot is synchronized to a common clock and becomes active at each time cycle or round.\\\\
\textbf{Time and Memory Complexity: }We evaluate the time in terms of the number of discrete rounds or cycles before achieving {\sc Dispersion}. Memory is the number of bits of storage required by each robot to successfully execute {\sc Dispersion}. Our goal is to solve {\sc Dispersion} using optimal time and memory.
\\\\

Given a simple, anonymous, port-labelled, connected graph $G$ with $n$ memory-less nodes and $m$ edges with maximum degree $\Delta$. Consider a team of $k\leq n$ mobile robots residing arbitrarily on the nodes of the graph, with some of the robots being prone to crashes. We want to devise an algorithm such that each robot, which eventually remains active (some of the robots can crash during the algorithm), re-positions itself to a distinct node of $G$ and remain stationary thereafter i.e., no two active robots occupy a single node at the conclusion of the algorithm.  Below, we formally state the problem of fault-tolerant dispersion.
\begin{definition}[Fault-Tolerant Dispersion]
\label{def:FT-dispersion}
Given $k \leq n$ robots, up to $f$ of which are faulty (which may fail by crashing), initially placed arbitrarily on a graph of $n$ nodes, the non-faulty robots, i.e., the robots which are not yet crashed must re-position themselves autonomously to reach a configuration where each node has at most one (non-faulty) robot on it and subsequently terminate.
\end{definition}

\section{Our Results}

We consider a team of $k\leq n$ mobile robots placed on an arbitrary, undirected simple graph, consisting of $n$ anonymous, memory-less nodes and $m$ edges. The ports at each node are labelled. The robots have unique IDs and a restricted amount of memory (measured in number of $bits$). These robots have some computing capability and can communicate with other robots, only when they are at the same node. We consider two different starting scenarios, based on the initial configuration of the robots. When the robots start from a single node, we call the configuration $rooted$, otherwise, we call it an $arbitrary$ configuration. We further assume that $f\leq k$ faulty robots in the network are prone to crash at any point. Our first algorithm for the rooted configuration crucially uses depth first search (DFS) traversal and improves the round complexity from $O(f\cdot\text{min}(m,k\Delta))$ \cite{PS021} rounds to $O(k^2)$. The second algorithm for the arbitrary configuration is an entirely new result whose complexity depends upon the factors: the number of faulty robots $(f)$, number of robot clusters $(l)$, total number of edges in the graph $(m)$, number of robots $(k)$ and the highest degree of the graph $(\Delta)$. In this case, the round complexity is  $O((f+l)\cdot \text{min}(m, k \Delta, k^2))$. The results are summarized in the following two theorems:


\medskip
\noindent \textbf{Theorem }\ref{thm: single-source}\textbf{(Crash Fault with Rooted Initial Configuration)} \textit{Consider any rooted initial configuration of $k \leq n$ mobile robots, out of which $f \leq k$ may crash, positioned on a single node of an arbitrary, anonymous $n$-node graph $G$ having $m$ edges, in synchronous setting {\sc Dispersion} can be solved deterministically in $O(k^2)$ rounds with $O(\log (k+\Delta))$ bits memory at each robot, where $\Delta$ is the highest degree of the graph.}

\medskip
Theorem \ref{thm: single-source} improves over the previously known algorithm (in the worst case, improvement is from cubic to quadratic) that takes $O(f\cdot\min(m,k\Delta))$ rounds for $f$ faulty robots \cite{PS021}. The theorem also matches the optimal memory bound ($\Omega(\log(\text{max}(k,\Delta)))$ \cite{KMS19}) with $O(\log(k+\Delta))$ bit memory and can handle any number of crashes. 

\medskip
\noindent \textbf{Theorem }\ref{thm: clustered}\textbf{(Crash Fault with arbitrary Initial Configuration)} \textit{Consider any arbitrary initial configuration of $k \leq n$ mobile robots,  out of which $f \leq k$ may crash and  positioned on $l \leq k/2$ nodes of an arbitrary and anonymous $n$-node graph $G$ having $m$ edges, in synchronous setting {\sc Dispersion} can be solved deterministically in $O((f+l)\cdot\text{min}(m,k\Delta, k^2))$ time with  $O(\log(k+\Delta))$ bits memory at each robot.}

\medskip
Theorem \ref{thm: clustered} solves the dispersion for arbitrary configuration with optimal memory per robot. The time complexity matches the one conjectured by Pattanayak \etal~\cite{PS021}. When $f,l$ and $\Delta$ are constants, the time complexity matches the lower bound of $\Omega(k)$. Moreover, the algorithm can handle any number of faulty robots. The results are summarized in the Table~\ref{tab:results}.

\section{Crash-Fault Dispersion from Rooted Configuration}\label{sec: rooted_configuration}
In this section, we present a deterministic algorithm that disperses the robots with single-source (rooted configuration) in adaptive crash fault. Our goal is to minimise the round complexity as well as keep the memory of the robots low (in bits). 
\vspace{-0.5cm}
\subsection{Algorithm}
In the absence of faulty nodes, one can run the DFS (depth first search) algorithm to solve the robot dispersion problem in $O(\textit{min}(m, k\Delta))$ rounds. But in the presence of crash faults setup, due to crashes, it becomes challenging to explore the graph. Classic dispersion algorithms rely on the robots themselves to keep track of the paths during exploration. The presence of a crashed robot in this instance may lead to an endless cycle. Therefore, our goal is  to ensure the dispersion of mobile robots despite the presence of faulty robots. 

In the rooted configuration, to manage the presence of faults, we avoid exploring the graph together with all the robots. That is, the graph is explored sequentially such that each robot $r_i$ ($1 \leq i \leq k$) does not begin exploring the graph, until the previous robot $r_{i-1}$ is guaranteed to have settled. During exploration, whenever a robot $r_i$ finds an empty node it settles down at that spot. Let us call this algorithm as {\sc Rooted-Crash-Fault-Dispersion}. Below, we explain the algorithm in detail.\\

\noindent \textbf{Functionality:} For simplicity, let us assume that the robots' ID lies in the range of $[1,\,k]$. Otherwise, the robots can map their IDs from the actual range to the range $[1,\,k]$, since the IDs are distinct. We denote the rooted configuration by $R_c$. We slightly abuse notation and use $R_c$ to indicate both the root and the initial gathering of robots. Robots at $R_c$ traverse the graph via DFS (Depth First Search) approach, where the decision of which edge to traverse first is based on the port numbers. The process proceeds in increasing order of IDs, starting with the robot with the minimum ID at $R_c$. $R_c$ then sends each robot to explore the graph via DFS.

Let the robot with the current minimum ID be $r_i$. Then $r_i$ begins to explore the graph via DFS (starting with the minimum port number at $R_c$). Once it leaves $R_c$, it has $3i$ rounds within which it can either i) settle at the first empty node it finds or ii) return to $R_c$ if it does not find an empty node to settle within $2i$ rounds.  If $r_i$ reports to $R_c$ within $3i$ rounds, then $R_c$ ensures that it does not release the robot with the next lowest ID, say  $r_{i+1}$.
This can be guaranteed as $r_i$ needs to traverse at most $(i-1)$ edges to explore the sub-graph traversed by $r_{i-1}$. $r_i$ requires at most $i$ rounds to return to the base $R_c$ since the next traversed edge might lead to the already visited node which is not empty. As $r_i$ requires $i$ rounds to report at the $R_c$, therefore, $r_i$ explores the graph for only $2i$ rounds. Notice that a robot will not traverse at the distance of more than $(i+1)$, before that, there will be an empty edge at a distance (distance from the root) of $(i+1)$ and the robot will settle down there. If $r_i$ did not find the empty node within $2i$ rounds then it starts to traverse towards $R_c$. In this way, $r_i$ reports to $R_c$ within $3i$ rounds so that $R_c$ does not send another robot to explore the graph. $R_c$ re-sends $r_i$ to explore the graph.
In this way, any $r_i$ traverses the graph until it finds an empty node. Note that in our process, we ensure that there are no two robots that are exploring the graph at the same time. 

To maintain the protocol, each $r_i$ maintains the following fields. Its ID $(r_i)$, a parent pointer ($r_i.parent$) that represents the edge it traversed,  a current direction pointer $(r_i.cdr)$ which indicates the direction it is required to follow. And finally, a backward traversal value $(r_i.B)$ which is initially $0$, and is set to $1$ once the backward traversal is complete. Here, our procedure performs the traditional DFS protocol but one-by-one, that is, the robots do not explore the graph simultaneously. In the following subsection, we give a detailed procedure the DFS procedure.

\subsubsection{DFS Traversal Procedure}\label{DFS Procedure}
{
\noindent Let the robots positioned initially at root $R_c$ be denoted by $R_c=\{r_1,r_2,\dots,r_k\}$, where $r_i$ is robot with ID $i$. For the \textit{rooted configuration}, each robot stores the following four variables :
    \begin{enumerate}
            \item $r_i.parent:$ the port number through which the robot $r_i$ has entered a new empty node and has settled. Initially, it is assigned to $null$.
            \item $r_i.cdr:$ the current direction of a settled robot $r_i$. A settled robot sets $cdr$ as the minimum available port number, however, as the algorithm progresses and more nodes are explored the robot $i$ can change its $cdr$ value, if needed.
            \item $r_i.B:$ a binary variable denoting backtrack status, initially assigned to $0$, takes the value $1$, if and only if, every sub-graph accessible through each of $r_i$'s child edges has been explored and now, no new part of $G$ could be explored through the node containing $r_i$.
            \item $r_i.settled:$ a binary variable, initially assigned $0$, takes the value $1$, if and only if, $r_i$ has settled at a particular node of $G$
    \end{enumerate}
\noindent\textbf{Update procedure:} In the first round, the robot $r_1$ assigns $r_1.settled\leftarrow1$ and sets $r_1.cdr\leftarrow1$, the minimum port number available at $R_c$. In the next round, the robot $r_2$ at $R_c$ reads and exits through $r_1.cdr$ to land at a new empty node (say $w$) and consequently sets $r_2.settled\leftarrow1$. Assume that $r_2$ arrived at $w$ through port $p_w$. $r_2$ writes $r_2.parent\leftarrow p_w$ and if $r_2.cdr \leq deg(w)$,  $r_2.cdr\leftarrow r_2.cdr+1$, if port $r_2.cdr+1\neq p_w$, else $r_2.cdr\leftarrow r_2.cdr+2$. In the third phase, the incoming robot $r_3$, following the $cdr$ pointers to $w$, can now decide to march forward or backtrack based on the following conditions:

\begin{itemize} \label{update}

    \item \textbf{forward} : $if$ ($p_w = r_2.parent$ or $p_w=$ old value of $r_2.cdr$) and (there is at least one port at $w$ that has not been taken yet). The robot $r_3$ exit $w$ through port $r_2.cdr$

    \item \textbf{backtrack} : $if$ ($p_w = r_2.parent$ or $p_w=$  old value of $r_2.cdr$) and (all the ports of $w$ have been taken already). Then, the robot $r_3$  exits $w$ through port $r_2.parent$. In such case, $r_3$ also sets the backtrack value of $r_2.B\leftarrow1$
    
\end{itemize}

There is another condition, denoting the onset of a cycle, under which choosing the backtrack phase is in order. When a robot $r_j$ enters $x$ through $p_x$ and robot $r$ is settled at $x$,
\begin{itemize}
    \item \textbf{backtrack} : $if$ ($p_x \neq r.parent$ and $p_x \neq $ old value of $r.cdr$). The robot $r_j$ exits $x$ through port $p_x$ and no variables of $r$ are altered.
\end{itemize}

Each robot explores the graph one by one and settles eventually. In the $ith$ phase, a robot $r_i$ has $3\cdot i$ number of rounds, in which, it can either settle down or explore $\frac{i}{2}$ new edges of $G$ (it takes i rounds for $r_i$ to come to $r_{i_1}$, takes additional $i$ rounds to explore $\frac{i}{2}$ new edges and return to $r_{i-1}$ and further $i$ rounds to return at $R_c$, following the $parent$ pointers). \\

\medskip
\noindent Now, coming back to the main algorithm, we now explain the \textbf{Decision} part.
\medskip\\
\noindent\textbf{Decision:} If $r_i$ encounters an unexpected child, $r_u$ i.e., a child whose parent and current pointer direction are set in the inappropriate direction w.r.t the perspective of $r_i$. It considers (correctly) that $r_u$ replaced a robot that has previously crashed. 
In such a situation, $r_i$ changes the parent of $r_u$ appropriately, i.e., the minimum available port number other than $r_u.parent$ (see, Figure~\ref{fig:rooted}). 

\begin{algorithm}[ht!]
    \caption{\sc Rooted-Crash-Fault-Dispersion}
    \begin{algorithmic}[1]
    \Require{An anonymous network of $n$ nodes where $f$ robots are faulty such that $f \leq k \leq n$. $k$ is the number of robots.}
    \Ensure{Robots' Dispersion.}
    \Statex
    \State Each robot $r_i$ maintains its ID, parent pointer, current direction pointer and the backward traversal pointer as $\langle r_i, r_i.parent, r_i.cdr, r_i.B \rangle$, respectively. Initially, $r_i.cdr$ is the same, i.e., minimum port number, $r_i.B=0$ for all the $k$ robots at rooted configuration. $r_i.B=1$ indicates that backtrack is done for that settled robot, while $r_i.B=0$ represents backtracking is remaining.
    
    \For{$7k^2$ rounds} \Comment{As stated in Lemma \ref{lem: exact_cost}}
        \State Each rooted configuration $(R_c)$ traverse the graph via DFS (Depth First Search) by sending the current minimum ID robot $(r_i)$ based on the current direction pointer at $R_c$ after every $3i$ round. If robot $r_i$ reports to $R_c$ within $3i$ rounds, then $R_c$ resends $r_i$ until there is no report regarding $r_i$.\Comment{DFS described in Section~\ref{DFS Procedure}}
        
        \State Each $r_i$ settles down at the first empty node it finds and sets all the attributes accordingly. \Comment{See functionality for more detail.}
        
         \If{$r_i$ does not find an empty node in $2i$ rounds}\label{line: report_3i}
            \State $r_i$ proceed to reports $R_c$. \Comment{$r_i$ reports within $3i$ rounds.}
        \Else
            \State $r_i$ settles at the empty node.
        \EndIf
        
        \If{$r_i$ encounters an unexpected parent for $r_u$ ($u<i)$}\Comment{See in decision.}\label{line: unexpected_encounter}
            \State $r_i$ resets the parent pointer, $r_u.parent$ and current direction pointer, $r_u.cdr$ based on minimum port available.
        \EndIf
    \EndFor
    \State All the non-faulty robots settle at a unique node. 
    \end{algorithmic}
\end{algorithm}\label{alg: rooted_config}

\begin{figure}[ht!]
\centering
  \includegraphics[scale=0.36]{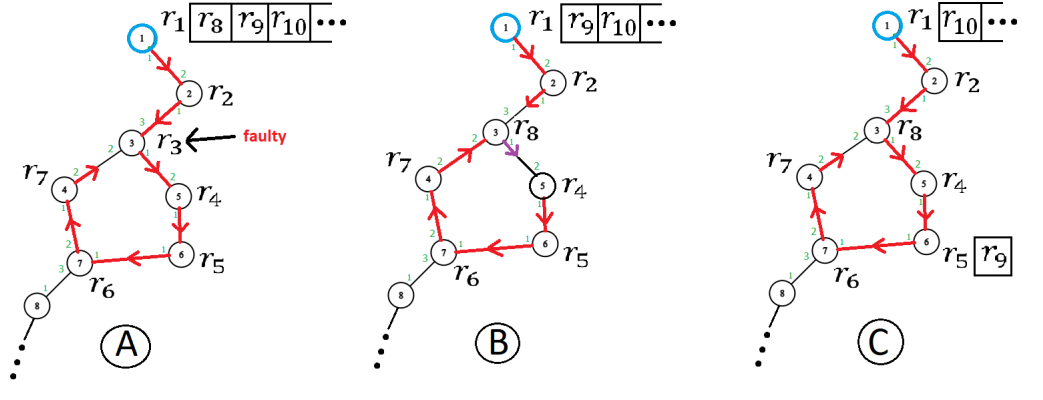}
  \caption{As seen in Fig. $A$, the robots start sequentially from root node $1$ (\textcolor{blue}{blue}) with robot $r_7$ settling and setting its $cdr$ $pointer$. In the next phase, while $r_8$ was at node $6$, $r_3$ crashed forcing $r_8$ to settle at node $3$, previously occupied by $r_3$ (Fig. $B$). Note that $r_8$ now marks $r_7$ as its $parent$. Such pointers are being \textit{corrected} by subsequent robot $r_9$ which has started from the root and now moving (Fig. $C$) . $r_9$ will subsequently meet $r_8$, then backtrack and settle at node $8$. 
  }\label{fig:rooted}
\end{figure}

\begin{lemma}\label{lem: non_faulty}
In the non-faulty setup, round complexity is $O(k^2)$.
\end{lemma}
\begin{proof}
 In a non-faulty setup, each robot behaves robustly and there are no crashes. Therefore, after the backtracking flag is set on a node, an edge is not traversed again during the DFS traversal. In traversing a graph from $R_c$, two kinds of situation may arise, either a robot $r_i$ reaches an empty node after $O(i)$ edge traversals, or it traverses $O(i^2)$ edges. In the first case, there is an empty node at a distance of $O(i)$. Therefore, $r_i$ settles at the empty node after $O(i)$ rounds. If such kind of situation arises repeatedly, then the algorithm takes $O(1)+O(2)+ \cdots + O(k) = O(k^2)$ rounds. In the second case, there might be a situation such that $r_i$ traverses $O(i^2)$ edges to find the empty node and only encounters previously settled nodes (at most $i(i-1)/2$ edges). 
 More preciously, $i/2$ new edges are traversed in $3i$ rounds. Notice that a robot will traverse only earlier traversed nodes at the distance $(i+1)$, if not,  then there will be an empty edge at a distance (distance from the root) of $(i+1)$ and the robot will settle down there. Therefore, $r_i$ covers $O(i^2)$ edges in $O(i^2)$ rounds  and future robots (i.e., robots having  ID $r_j$; $\forall$ $j>i$) will not traverse these edges again. Hence, we can conclude that non-faulty setup takes $O(k^2)$ rounds in the given model.
\end{proof}

\begin{lemma}\label{lem: extra_cost}
In the faulty setting, a crashed robot may bring about an extra cost of $O(k)$ rounds in comparison to the non-faulty setting.
\end{lemma}
\begin{proof}
In the faulty setup, a robot might crash at any time and the respective node becomes empty, say node $v_i$. As a consequence, the information held by that robot (at the node $v_i$) is also lost. Accordingly, the next robot that discovers $v_i$, say $r_i$, settles down at $v_i$. A robot possesses the information of current direction, parent node and backtracking status apart from its own ID. For that reason, the current direction pointer is pointing towards the edge based on its least labelled edge. But there might be the case (in the worst case) that the last crashed node has traversed up to $(i-2)$ edges which should be traversed again by the $r_{i+1}$. This takes extra $O(i)$ rounds. Also, in the worst case, this value can be $O(k)$ since the number of robots is $k$. Hence, the lemma.
\end{proof}

\begin{lemma}\label{lem: one_robot_moving}
There is (at most) one robot moving (neither settled at its respective node, nor at  rooted configuration $R_c$) at any instance.
\end{lemma}
\begin{proof}
Proof by contradiction, let us suppose there exist two robots in moving condition, say $r_i$ and $r_{i+1}$. Also, assume $r_i$ started before, $r_{i+1}$. Now, as $r_i$ has not settled,   $r_i$ reports to $R_c$ every $3i$ rounds (Line~\ref{line: report_3i} of algorithm {\sc Rooted-Crash-Fault-Dispersion}). But if $r_i$ reports every $3i$ rounds then $R_c$ does not release the next robot which is contradictory to our assumption.
\end{proof}

\begin{lemma}\label{lem: loop}
A loop or cycle may be formed by the current direction pointer ($cdr$ $pointer$). The algorithm {\sc Rooted-Crash-Fault-Dispersion} successfully avoids any loop during dispersion.   
\end{lemma}
\begin{proof}
During the execution of the algorithm, a loop or cycle may be formed if a robot $r_i$ crashes at a node $n_i$ then the current direction pointer ($cdr$ $pointer$) is set by the upcoming robot $r_{i+1}$ with the lowest port. That lowest port might have been traversed earlier. Therefore, a loop is formed (as shown in Figure~\ref{fig:rooted}, see A and B). From Lemma \ref{lem: one_robot_moving}, we know that only one robot is moving at any instance, say $r_{i+1}$. Therefore, $r_{i+2}$ (the next robot) starts after $r_{i+1}$ settles. If $r_{i+2}$ encounters any robot with an unexpected $cdr$ $pointer$ then $r_{i+1}$ changes the $cdr$ $pointer$ appropriately (Line \ref{line: unexpected_encounter} of the algorithm {\sc Rooted-Crash-Fault-Dispersion} and the Figure~\ref{fig:rooted}, see C). Thus, loops are avoided in the network.
\end{proof}


\begin{lemma}\label{lem: exact_cost}
The algorithm {\sc Rooted-Crash-Fault-Dispersion} takes at most $7k^2$ rounds and $O(\log (k+D))$ bits memory. 
\end{lemma}
\begin{proof}
In case of round complexity, a non-faulty set-up from Lemma \ref{lem: non_faulty}, the total number of rounds are $3(1+2+ \dots + k)<3k^2$ (in the best case where $r_i$ finds the empty node within $3i$ rounds). Additionally, a robot can traverse at most $i/2$ new edges in $3i$ rounds (in a particular phase) without settling down on an empty node (in the worst case). Therefore, round complexity for $k(k-1)/2$ edges in the non-faulty setup is $<3k^2$. Moreover, from Lemma \ref{lem: extra_cost}, we know that the extra cost incurred for $f$ robots' crashing is at most $fk$. Hence, overall round complexity is at most $3k^2+3k^2+k^2=7k^2$.

In case of memory complexity, each robot stores its ID which takes $O(\log k)$ bit space. Along with that parent pointer and current direction pointer takes $O(\log \Delta)$ bit memory each. While the backward pointer take a single bit. Therefore, the memory complexity is $O(\log(k+ \Delta))$. 
\end{proof}

From the above discussion, we conclude the following result.

\begin{theorem}\label{thm: single-source}
Consider any rooted initial configuration of $k \leq n$ mobile robots, out of which $f \leq k$ may crash, positioned on a single node of an arbitrary, anonymous $n$-node graph $G$ having $m$ edges, in synchronous setting {\sc Dispersion} can be solved deterministically in $O(k^2)$ time with $O(\log (k+\Delta))$ bits memory at each robot, where $\Delta$ is the highest degree of the graph.
\end{theorem}

\section{Crash-Fault Dispersion for Arbitrary Configurations}\label{sec: cluster}

In this configuration setting, the robots are distributed across the graph in clusters such that there are $C=\{C_1,\ldots,C_l\}$ groups of robots at $l$ different nodes at the start such that $\sum_{i} C_i=k$. The goal of the dispersion is to ensure that the robots are dispersed among the graph vertices such that each node has at most one robot. In this setting, we assume that the robots are aware of $k, f,\Delta, l$ and $m$.\\


\noindent\textbf{Procedure:} 
Our protocol runs in phases, in which each phase consists of $\text{min}(m,k\Delta, k^2)$ rounds. At the start of each phase, each cluster begins a $counter$ that counts down from $\text{min}(m,k\Delta, k^2)$. Each cluster $C_i$ then begins exploring the network simultaneously via the traditional DFS algorithm(in the trivial case of a  singleton cluster consisting of only one robot, it considers itself dispersed). Unlike the rooted configuration, individual robots do not explore and return, the entire cluster moves together. Whenever a cluster encounters a new (empty) node in the network, the robot with the current highest ID in the cluster settles, and sets its pointers appropriately. At the end of each round, the counter is decreased by $1$. When the counter becomes zero, it signals the end of the phase, and all flags are reset. That is all pointers become null, including the pointers of already settled robots.  
After that, each cluster starts exploring the network with its current node as a point of origin. This continues until all robots in the cluster settle. Details of this procedure can be found in the pseudocode ~\ref{alg: cluster_config}\\

\noindent\textbf{Detailed Procedure:} There are two main parts to the protocol, i) exploration, ii) encounter. Exploration deals with the general procedure involved in exploring the graph. While encounter deals with the details involved in robots from different clusters meeting.

Let's begin with all the information stored at a robot. Each robot $r$ in a cluster $C_i$ consists of the following pointers $cid$, $parent$, $cdr$, $priority$, and $B$ (backtrack). The pointer $cid$ denotes the ID of the cluster it belonged to when a robot settles. $cid$ of a cluster $C_i$ is determined at the start of the phase, and is the ID of the robot with the highest ID.  When a robot decides to settle at a node, the $parent$ pointer keeps track of the port through which it entered the node. Similarly, the $cdr$ pointer is used to keep track of the port through which a cluster leaves the node in which it is settled. The $priority$ pointer of a settled robot keeps track of its priority in various clusters, originally this is simply the $cid$ of the cluster it was part of, that is the priority of a cluster is simply its $cid$. However, a robot's priority may change if a higher priority cluster discovers it and updates its priority pointer. In our work, priority is decided by the cluster's ID, that is, between two clusters, the cluster with the higher $cid$ has higher precedence. And of course, the $B$ of the backtrack pointer keeps track of the backtrack status of its DFS. In addition to all of these, each robot also has a field called $counter$, which is set to $\text{min}(m,k\Delta, k^2)$ at the beginning of each phase. Note that since all robots set the counter at the beginning of the phase simultaneously, the counter has the same value across all robots.\\

           

\noindent\textbf{Exploration:} 
As mentioned before, as long as a cluster is non-empty, at the beginning of each phase, each cluster begins exploring the graph via the traditional DFS until the cluster is empty or it encounters a robot from a higher priority cluster (more on this in the encounter section). 
In each phase, each robot in any cluster $C_i$ sets its $cid$ and $priority$ to the highest ID in the cluster, and its counter to $\text{min}(m,k\Delta, k^2)$. We consider the node in which $C_i$ is at the start of the phase to be its root. $C_i$ then follows the traditional DFS format for exploration. It leaves the node via the smallest unexplored port. If the node is empty, the robot with the current minimum ID, say $r$, sets its \textit{parent} and $cdr$ pointers and settles at the node.

The update function for the $cdr$ pointer is exactly the same as the one in the rooted case, i.e., it follows the traditional DFS procedure, except that all the robots in the cluster move through the exit port. Here, the robots use a similar DFS traversal technique as explained in DFS Procedure~\ref{DFS Procedure}. In the arbitrary case, the whole group moves together except for the robots that have settled. Whenever an exit port is calculated for moving, every (unsettled) robot in the group moves out through the port into the neighbouring node. 
All robots in $C_i$ decrease their counter by one and $C_i$ leaves through the port in $r.cdr$. If a cluster ever finds itself returning to a node with a robot $r$ from its own cluster and it has exhausted all of the ports in which $r$ has settled, then it sets $r$'s backtrack flag. Once a phase has finished, if the cluster is non-empty  it resets all flags and counter and begins DFS once again. During exploration, if the cluster $C_i$ reaches a node $u$ whose degree is $k$,  then they use BFS to explore the neighbourhood of $u$ and settle the robots of $C_i$ in at most $O(k)$ rounds. However, here we have not explored what happens if a robot from a cluster $C_i$ meets a robot from $C_j$. That brings us to the next important part of the protocol, the encounters. Details of the exploration part of a cluster can be found in Algorithm~\ref{alg:explore}.\\

\noindent\textbf{Encounter:} This section contains the explanation of the encounter part of the protocol (see Algorithm~\ref{alg:encounter}). When a robot (or cluster) meets, that is \textit{encounters} a robot from a different cluster, the next step in exploration is decided based on priority. Simply put, the robot with higher priority always takes precedence as follows. There are two distinct scenarios, i) a cluster finds an already settled node ii) multiple clusters meet on the same node.
In the first case, if a cluster with a higher priority (say $C_i$) finds a robot $r_p$ from a lower priority cluster (say $C_j$) on a node, it sets $r_p$'s priority to its own priority, resets $r_p$'s parent and $cid$ to its own, and finally sets $r_p$'s $cdr$ (to the minimum unexplored port the higher priority cluster has not explored so far) and continues its DFS. If on the other hand, a lower priority cluster finds a robot $r_p$ with a higher priority, it stops its exploration and just continues decreasing its counter at every round till the end of the phase, and begins the exploration in the next phase. Note, if a cluster finds a settled robot whose flags have been reset (i.e., set to $null$), then it's the same scenario as that of finding a robot from a lower priority cluster. The settled robot takes the priority and ID of the newly arrived cluster.

In the second scenario, if two (or more) clusters meet, the clusters merge and take on the priority of the cluster with the highest priority among them. They stop and countdown and begin exploration as a merged cluster in the next phase. See Figure~\ref{fig:cluster} for an illustration of various kinds of encounters.
           

{Note that the number of clusters is non-increasing between two consecutive phases. At any phase, a cluster may either (i) disperse over the nodes completely, or (ii) survive to explore in the next phase, or (iii) merge with a higher priority cluster. Thus, the number of clusters either remains the same or decreases at the end of every phase. Now we show that after ($l+f$) phases, dispersion is achieved.}

\begin{algorithm}[H] 
    \caption{\sc Arbitrary-Crash-Fault-Dispersion}
    \begin{algorithmic}[1]
    \Require{An anonymous network of $n$ nodes where $f$ robots are faulty such that $f \leq k \leq n$. $k$ is the number of robots. The robots are distributed across the network in $C=\{C_1,C_2,\ldots,C_l\}$ clusters such that $\sum_{i} |C_i|\leq k$. The value of $f, l, m, k$ and $\Delta$ are known to the robots.}
    \Ensure{Robots' Dispersion.}
    \Statex Each robot $r$ belonging to a cluster $C_i$ maintains the following pointers. $r.cid$, $r.priority$, $r.parent$, $r.cdr$ and $r.counter$. Initially $r.counter$ is set to $\text{min}(m,k\Delta,k^2)$,where $\Delta$ is the maximum degree of the network.
    \Statex $j=0$.
    \While{$j \leq (l+f)$}
    \State $counter=\text{min}(m,k\Delta,k^2)$
         \While {$counter> 0$}
           
             \State All non-empty clusters $C_i$ perform Explore($C_i$).
             \State $counter=counter-1$.
            \EndWhile
        \State Reset all pointers to null. Set $r.counter=\text{min}(m,k\Delta, k^2)$ across all $r$.
        \State $j=j+1$.
    \EndWhile    
    \end{algorithmic}
\end{algorithm}\label{alg: cluster_config}

\begin{algorithm}[htp]
    \caption{\sc Encounter($C_i$)}
    \begin{algorithmic}[1]
    \Require{A non-empty cluster of robots $C_i$}
    
    \If{node is not empty and contains a robot $r_p$}
             
        \If{$r_{p}.cid=C_{i}.cid$}
             \If {$r_{p}.B$ is set}
             \State Return through $r.parent$.
             \Else
              \State Continue to explore. Update $r_p.cdr$ to minimum unexplored port. Move through $r_{p}.cdr$. \Comment{See Section~\ref{update} for detailed procedure.}
             \EndIf 
           
        \EndIf
        \If{$r_{p}.cid \neq C_{i}.cid$}
          \If{$r_p$ has higher priority}
          \State Wait for $counter$ to become zero.
          \Else
          \State $r_{p}.cid=C_{i}.cid$.
          \State $r_{p}.priority=C_{i}.priority$.
          \State $r_p.parent=$ port through which $C_i$ entered.
          \State Continue to explore. Update $r_p.cdr$ to minimum unexplored port. Move through $r_{p}.cdr$. \Comment{See Section~\ref{update} for detailed procedure.}
          \EndIf  
         \EndIf  
         
        \If{$r_{p}.cid=null$} \Comment{$r_p$'s flags have been reset.}
          \State $r_{p}.cid=C_{i}.cid$.
          \State $r_{p}.priority=C_{i}.priority$
          \State $r_p.parent=$ port through which $C_i$ entered.
          \State $r_{p}.cdr$ is set to the minimum unexplored port and the cluster then moves through $r_{p}.cdr$. \Comment{See Section~\ref{update} for detailed procedure.}
         \EndIf 
    \EndIf
    \If{node is not empty and contains clusters $C_s \subset C$ that are not $C_i$}
      \If{$C_i$ has highest priority}
         \State Explore ($C_i$).
        \Else 
          \State Let $C_j$ be the cluster on the node with the highest priority.
          \State $C_i$ merges with all remaining clusters in $C_s/C_j$. It takes the priority of the highest priority cluster in $C_s/C_j$.
          \State The clusters wait until the start of the next phase to begin exploration.
        \EndIf
    \EndIf    
    \end{algorithmic}
\end{algorithm}\label{alg:encounter}

\begin{algorithm}[t]
    \caption{\sc Explore($C_i$)}
    \begin{algorithmic}[1]
    \Require{A non-empty cluster of robots $C_i$}
    \If{node is empty}
    \State Settle robot with current lowest ID in $C_i$ (say $r$).
    \State $r.priority \leftarrow r.cid$.
    \State $r.counter\leftarrow counter$.
    \State $r.parent \leftarrow$ port through which $r$ entered the node.
    \State $r.cdr$ is the minimum port that hasn't been explored so far.
    \EndIf
    \If{node is not empty}
           \State Perform Encounter($C_i$).
           
    \EndIf
    \end{algorithmic}
\end{algorithm}\label{alg:explore}

\begin{figure}[!ht]
  \centering
    \includegraphics[scale=0.36]{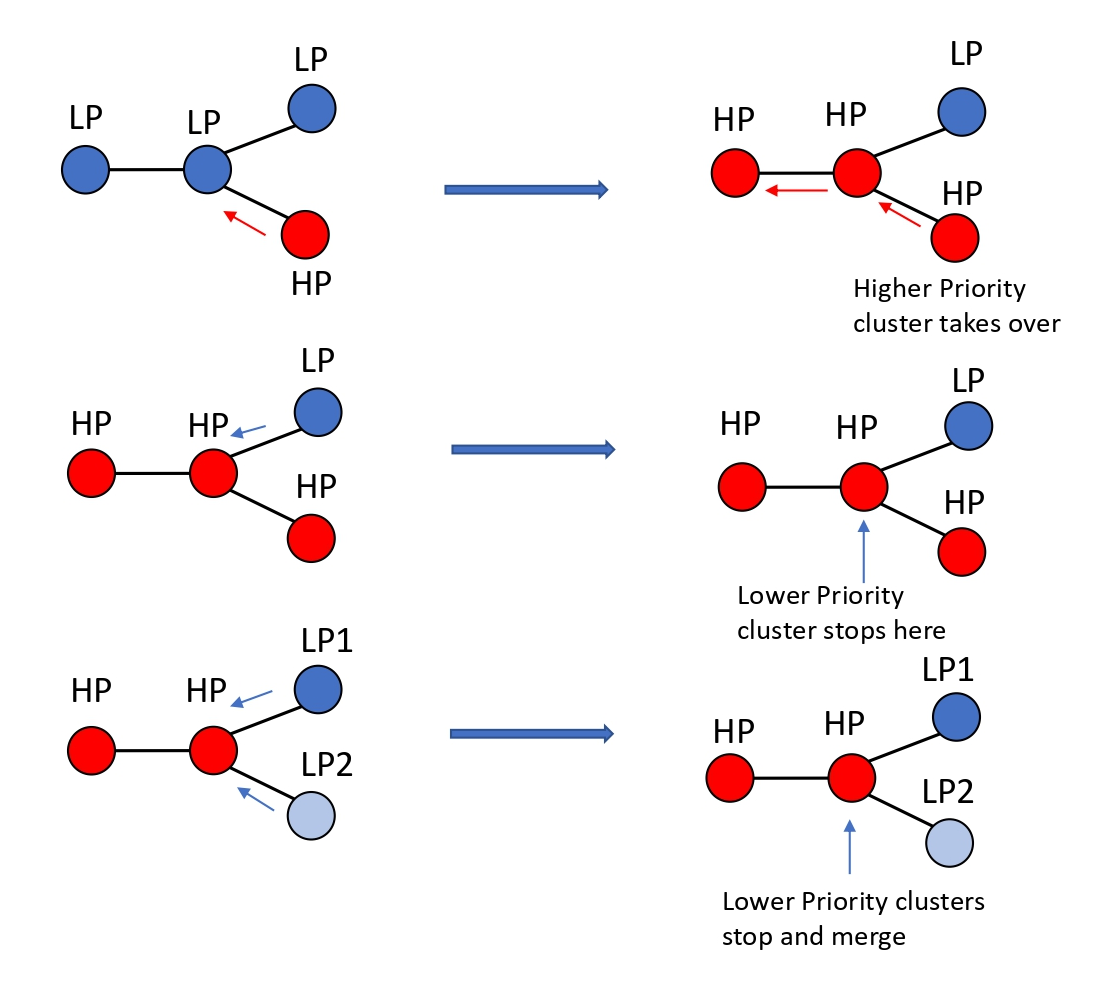}
  \caption{Illustration of various types of encounters between clusters of different priority. The arrows represent the direction the clusters are moving. Higher priority clusters are colored red and lower priority clusters blue. When a higher priority cluster meets a robot from a cluster with lower priority, it resets the pointers of the lower priority robot. Meanwhile, anytime a lower priority cluster meets a robot from a higher priority cluster, it stops and waits for the end of the phase. }\label{fig:cluster}
  \end{figure}

\begin{lemma}\label{lem: crash}
The effects of a robot crash, that is time delay caused by the presence of a crash are limited to the phase it occurs in. After that, it ceases to have an effect.
\end{lemma}
\begin{proof}
Since at the end of every phase, all robots reset their flags, including the parent and $cdr$ pointers, previously explored paths are equivalent to new unexplored paths in the current phase, as their pointers are set by the currently exploring clusters. Hence, previous phases do not have any impact on the DFS running in the current phase.
 \end{proof}

\begin{lemma}\label{lem: dispersehp}
Let $C_i$ be the cluster with the highest priority in Phase $j$. $C_i$ is guaranteed dispersion by the end of $j$ if $j$ is fault-free.
\end{lemma}
\begin{proof}
From Lemma~\ref{lem: crash} we know that crashes in previous rounds do not have an effect on exploration in the current phase. And, in the absence of faults during the phase itself, we see that $C_i$ exploration is equivalent to a rooted single cluster exploration of the network presented in~\ref{update}. Thus it is able to complete its dispersion using DFS without any delays or interference from other clusters, which takes less than $O(\text{min}(m,k\Delta, k^2))$ rounds to complete. 
\end{proof}

\begin{lemma} \label{lem: fault-free}
Each cluster $C_i\in C$ is guaranteed to have at least a single fault-free phase in which it has the highest priority. 
\end{lemma}

\begin{proof}
Quite trivially, since there are $(l+f)$ phases, each cluster is guaranteed at least one phase in which no faults occur, and in which they are the highest priority.  
\end{proof}
\begin{lemma}\label{lem:l+f}
At the end of $(l+f)$ phases, all clusters are guaranteed to have dispersed.
\end{lemma}
\begin{proof}
This follows directly from Lemmas~\ref{lem: dispersehp} and~\ref{lem: fault-free}. Each cluster is guaranteed to have at least one fault-free phase in which it has the highest priority. From~\ref{lem: crash} we know in that phase there is guaranteed dispersion. 
Hence, in $(l+f)$ phases, we are guaranteed to have total dispersion of all clusters.
\end{proof}
Thus, we have the following theorem.
\begin{theorem}\label{thm: clustered}
In the synchronous setting, the crash-tolerant algorithm for the arbitrary configuration (algorithm {\sc Arbitrary-Crash-Fault-Dispersion}) ensures dispersion of mobile robots in an arbitrary graph from an arbitrary initial configuration in $O((f+l)\cdot\text{min}(m,k\Delta, k^2))$ rounds with each robot requiring $O(\log (k+\Delta))$ bits of memory.
\end{theorem}
Since the number of clusters $l$ and the number of faulty robots $f$ are both less than the total number of robots $k$, we have the following remark. 
\begin{remark}
If only the number of robots $(k)$ is known and all other factors are unknown to the network then the algorithm for arbitrary configuration takes $O(k^3)$ rounds.
\end{remark}

\section{Conclusion and Future Work}\label{sec: conclusion}
In this paper, we studied Dispersion for distinguishable mobile robots on anonymous port-labelled arbitrary graphs under crash faults. We presented a deterministic algorithm that solves robot dispersion in two different settings, i) with a rooted configuration of robots and ii) an arbitrary configuration of robots. We achieved the $O(k^2)$ round complexity in rooted configuration while $O((f+l)\text{min}(m,k\Delta, k^2))$ round complexity in arbitrary setting. In both cases, we used $O(\log (k+\Delta))$ memory. Some open questions that are raised by our work: i) What is the non-trivial lower bound for the round complexity in both the setting by keeping the memory $O(\log (k+\Delta))$? ii) Is it possible to give a similar round complexity for the case of arbitrary configuration as we achieved in rooted configuration? iii) Is it possible to get the same bound in the arbitrary configuration without the knowledge of $f, l, \Delta$ and  $m$? iv) Finally, whether similar bounds hold in the presence of Byzantine failures?

\bibliographystyle{plain}\small
\bibliography{references}

\begin{thebibliography}{10}

\bibitem{AAMKS18}
Ankush Agarwalla, John Augustine, William~K. {Moses Jr.}, Sankar {Madhav K.},
  and Arvind~Krishna Sridhar.
\newblock Deterministic dispersion of mobile robots in dynamic rings.
\newblock In {\em {ICDCN} 2018.}

\bibitem{AM18}
John Augustine and William~K. {Moses Jr.}
\newblock Dispersion of mobile robots: {A} study of memory-time trade-offs.
\newblock In {\em {ICDCN} 2018.}

\bibitem{BGHIKK09}
Evangelos Bampas, Leszek Gasieniec, Nicolas Hanusse, David Ilcinkas, Ralf
  Klasing, and Adrian Kosowski.
\newblock {Euler Tour Lock-in Problem in the Rotor-Router Model}.
\newblock In {\em {DISC 2009.}}

\bibitem{CFIKP08}
Reuven Cohen, Pierre Fraigniaud, David Ilcinkas, Amos Korman, and David Peleg.
\newblock Label-guided graph exploration by a finite automaton.
\newblock {\em ACM Trans. Algorithms}, 2008.

\bibitem{DBS21}
Archak Das, Kaustav Bose, and Buddhadeb Sau.
\newblock Memory optimal dispersion by anonymous mobile robots.
\newblock In {\em {CALDAM} 2021.}

\bibitem{DDKPU13}
Dariusz Dereniowski, Yann Disser, Adrian Kosowski, Dominik Pajak, and
  Przemyslaw Uznanski.
\newblock {Fast Collaborative Graph Exploration}.
\newblock In {\em {ICALP 2013.}}

\bibitem{FIPPP05}
Pierre Fraigniaud, David Ilcinkas, Guy Peer, Andrzej Pelc, and David Peleg.
\newblock {Graph Exploration by a Finite Automaton}.
\newblock {\em {Theoretical Computer Science}}, 2005.

\bibitem{HABFM02}
Tien{-}Ruey Hsiang, Esther~M. Arkin, Michael~A. Bender, S{\'{a}}ndor~P. Fekete,
  and Joseph S.~B. Mitchell.
\newblock Algorithms for rapidly dispersing robot swarms in unknown
  environments.
\newblock In {\em {WAFR} 2002.}

\bibitem{HABFM03}
Tien{-}Ruey Hsiang, Esther~M. Arkin, Michael~A. Bender, S{\'{a}}ndor~P. Fekete,
  and Joseph S.~B. Mitchell.
\newblock Online dispersion algorithms for swarms of robots.
\newblock In Steven Fortune, editor, {\em SCG 2003.}

\bibitem{KA19}
Ajay~D. Kshemkalyani and Faizan Ali.
\newblock Efficient dispersion of mobile robots on graphs.
\newblock In {\em {ICDCN} 2019.}

\bibitem{KMS20}
Ajay~D. Kshemkalyani, Anisur~Rahaman Molla, and Gokarna Sharma.
\newblock Dispersion of mobile robots on grids.
\newblock In {\em {WALCOM} 2020.}

\bibitem{KMS19}
Ajay~D. Kshemkalyani, Anisur~Rahaman Molla, and Gokarna Sharma.
\newblock Fast dispersion of mobile robots on arbitrary graphs.
\newblock In {\em {ALGOSENSORS} 2019.}

\bibitem{KMS22}
Ajay~D. Kshemkalyani, Anisur~Rahaman Molla, and Gokarna Sharma.
\newblock Dispersion of mobile robots using global communication.
\newblock {\em J. Parallel Distributed Comput.}, 2022.

\bibitem{KMS020}
Ajay~D. Kshemkalyani, Anisur Rahaman~Molla, and Gokarna Sharma.
\newblock Efficient dispersion of mobile robots on dynamic graphs.
\newblock In {\em ICDCS 2020.}

\bibitem{KS21}
Ajay~D. Kshemkalyani and Gokarna Sharma.
\newblock Near-optimal dispersion on arbitrary anonymous graphs.
\newblock In {\em {OPODIS} 2021.}

\bibitem{LFPPSV20}
Giuseppe Antonio~Di Luna, Paola Flocchini, Linda Pagli, Giuseppe Prencipe,
  Nicola Santoro, and Giovanni Viglietta.
\newblock Gathering in dynamic rings.
\newblock {\em Theor. Comput. Sci.}, 2020.

\bibitem{MMM20}
Subhrangsu Mandal, Anisur~Rahaman Molla, and William K.~Moses Jr.
\newblock Live exploration with mobile robots in a dynamic ring, revisited.
\newblock In {\em {ALGOSENSORS} 2020.}

\bibitem{MM22}
Anisur~Rahaman Molla and William K.~Moses Jr.
\newblock Dispersion of mobile robots.
\newblock In {\em {ICDCN} 2022.}

\bibitem{MMM21}
Anisur~Rahaman Molla, Kaushik Mondal, and William~K. {Moses Jr.}
\newblock Byzantine dispersion on graphs.
\newblock In {\em {IPDPS} 2021.}

\bibitem{MMM020}
Anisur~Rahaman Molla, Kaushik Mondal, and William~K. {Moses Jr.}
\newblock Efficient dispersion on an anonymous ring in the presence of weak
  byzantine robots.
\newblock In {\em {ALGOSENSORS} 2020.}

\bibitem{MollaM19}
Anisur~Rahaman Molla and William~K. {Moses Jr.}
\newblock Dispersion of mobile robots: The power of randomness.
\newblock In {\em {TAMC} 2019.}

\bibitem{PS021}
Debasish Pattanayak, Gokarna Sharma, and Partha~Sarathi Mandal.
\newblock Dispersion of mobile robots tolerating faults.
\newblock In {\em {ICDCN} 2021.}

\bibitem{SSKM20}
Takahiro Shintaku, Yuichi Sudo, Hirotsugu Kakugawa, and Toshimitsu Masuzawa.
\newblock Efficient dispersion of mobile agents without global knowledge.
\newblock In {\em {SSS} 2020.}

\end{thebibliography}

\end{document}